\documentclass[journal]{IEEEtran}
\usepackage[latin1]{inputenc}
\usepackage{amsmath}
\usepackage{graphicx}
\usepackage{amssymb}
\usepackage{xcolor}
\makeatletter

\newtheorem{corol}{Corollary}

\newtheorem{thm}{Theorem}

\providecommand{\LyX}{L\kern-.1667em\lower.25em\hbox{Y}\kern-.125emX\@}

\def\BibTeX{{\rm B\kern-.05em{\sc i\kern-.025em b}\kern-.08em
    T\kern-.1667em\lower.7ex\hbox{E}\kern-.125emX}}

\makeatother
\begin{document}
\def\ZZ{{\mathbb Z}}
\def\RR{{\Bbb R}}
\def\NN{{\mathbb N}}
\def\CC{{\mathbb C}}

\title{Delayed Functional Observers for Output-Delayed Linear Systems}
	\author{Hieu Trinh
	\thanks{Hieu Trinh is with the School of Engineering, Deakin University, Waurn Ponds, 75 Pigdons Road, Geelong, Australia. (email:  hieu.trinh@deakin.edu.au)}
	}

\maketitle \maketitle \maketitle \thispagestyle{plain}
\pagestyle{plain}

\begin{abstract}

This paper introduces a novel class of delayed functional observers specifically designed to reconstruct delayed control laws under severe output measurement lags, directly complementing recent literature \cite{trinhnn26, trinhnam26}. By systematically mitigating simultaneous, unequal delays across both the actuator and sensor channels, the proposed architecture resolves dual-channel latency without requiring full-state estimation or computationally intensive real-time distributed integration. Ultimately, this work provides a powerful, low-order framework that bridges the gap between idealized control theory and the practical constraints of modern networked engineering systems.
\end{abstract}


\section{System Description and Problem Statement}

We consider the following linear system with mismatched time delays in both control input and output vectors
\begin{align}
	\label{b1}
	\dot{x}(t)&=Ax(t)+Bu(t-\tau_u),\\
		\label{b2}
	y(t)& = Cx(t-\tau_y),
\end{align}
where $x(t)\in \mathbb{R}^n$ is the state vector, $u(t)\in \mathbb{R}^r$ is the control input vector, and $y(t)\in\mathbb{R}^{p}$ is the output vector. $\tau_u>0$ and $\tau_y>0$ are the constant time delays in the input and output vectors, respectively. Matrices
$A\in\mathbb{R}^{n\times n}$, $B\in\mathbb{R}^{n\times r}$ and $C\in\mathbb{R}^{p\times n}$ are constant. Without loss of generality, let $B$ be a matrix of full column rank and
$C$ is a matrix of full row rank.

In this paper, we consider the case where $\tau_y > \tau_u$ without assuming that $A$ is Hurwitz. For this system, stabilization is achieved by employing the following delayed control law \cite{trinhnn26}
\begin{align}
	\label{b3}
	u(t-\tau_u)&=Fx(t-\tau_u),
\end{align}
where $F \in \mathbb{R}^{r \times n}$ is the control gain.

Since neither the current state vector $x(t)$ nor the input-delayed state vector $x(t-\tau_u)$ is available for feedback, we utilize the heavily delayed output vector (\ref{b2}). The measurements in this output vector consist of a subset of the state variables subject to a larger delay, $\tau_y > \tau_u$. Using this delayed output vector, we design an asymptotic observer to estimate the control law (\ref{b3}). The estimated functional is then implemented within an observer-based control scheme to stabilize the closed-loop system, despite the presence of distinct time delays in both the input and output channels.

This paper is organized as follows. Section II presents the design of functional observers to estimate the control law (\ref{b3}) as well as numerical examples to illustrate the results. Concluding remarks are given in Section III.

\section{Functional Observers for Estimating Delayed Control Laws}
In this section, we design functional observers to asymptotic estimate the delayed control law (\ref{b3})
\[z(t)=u(t-\tau_u)=Fx(t-\tau_u)\] using the heavily delayed output vector (\ref{b2}).

Let us first consider the following observer
\begin{align}
	\label{b5}
	\hat{z}(t)&= w(t)+My(t),\\
	\label{b6}
	\dot{w}(t)&=Nw(t)+N_{\tau}w(t-\tau)+Gy(t)+G_1y(t-\tau)  \nonumber\\
	&+ Ju(t-2\tau_u)+J_1u(t-\tau_u-\tau_y), 
\end{align}
where $\tau=\tau_y-\tau_u>0$, $w(\theta)=\rho(\theta)$  for $\theta\in[-\tau,0]$, and $w(t)\in\mathbb{R}^r$.
The matrices $M$, $N$, $N_{\tau}$, $G$, $G_1$, $J$, and $J_1$ are design parameters to be determined such that $\hat{z}(t)\to z(t)$ asymptotically. 

Defining the estimation error vector $e(t)=\hat z(t)-z(t)$, and using $\tau=\tau_y-\tau_u$, the error dynamics are given by
\begin{align}
	\label{b7}
	\dot{e}(t)&=\dot{w}(t)+M\dot{y}(t)-F\dot{x}(t-\tau_u)\nonumber\\ &=Ne(t)+N_{\tau}e(t-\tau)+\mathcal{C}_{1}x(t-\tau_u)+\mathcal{C}_{2}x(t-\tau_y)\nonumber\\ &+\mathcal{C}_{3}x(t-\tau-\tau_y)+\mathcal{C}_{4}u(t-2\tau_u)+\mathcal{C}_{5}u(t-\tau_u-\tau_y),
\end{align}
where\\ 
$\mathcal{C}_1 =NF-FA$, \ $\mathcal{C}_2 =N_{\tau}F+\bar{G}C+MCA$, \ $\mathcal{C}_3 =\bar{G}_1C$, \ $\mathcal{C}_4 =J-FB$, \ $\mathcal{C}_5 =J_1+MCB$, \ $\bar{G}:=G-NM$, $\bar{G}_1:=G_1-N_{\tau}M$, and $\mathcal{\bar{C}}=\begin{pmatrix}
	\mathcal {C}_1 &
	\mathcal {C}_2 & \mathcal {C}_3 &\mathcal {C}_4 &\mathcal {C}_5
\end{pmatrix}$.

\textit{Remark 1:} The existence conditions for the proposed observer \eqref{b5}-\eqref{b6} are identical to those required for the observer in \cite{trinhnn26} (specifically, equations (11)-(12)), which are themselves identical to the conditions established in Theorem 1 of \cite{trinhnam26}. However, the proposed observer (\ref{b5})-(\ref{b6}) offers distinct advantages over the design in \cite{trinhnn26}. By directly accounting for the net difference between output and input delays, this approach ensures more robust error stability while simplifying the structure with fewer delay units. Specifically, the stability analysis of the observer \eqref{b5}-\eqref{b6} reduces to ensuring the asymptotic stability of the following linear delay differential equation
\begin{align}
	\label{b8}\dot{e}(t) = Ne(t) + N_{\tau}e(t - \tau),
\end{align}
which is less conservative than its counterpart in \cite{trinhnn26} (specifically, equation (14) where $\tau < \tau_y$). Furthermore, a more robustly stable error system \eqref{b8} provides a faster convergence rate, thereby improving the performance of the observer-based closed-loop control system.

\textit{Remark 2:} In the case where $\text{rank}(C) = n$, a functional observer with fewer parameters can be constructed as follows
\begin{align}
	\label{b9}
	\dot{\hat{z}}(t)=N\hat{z}(t)+N_{\tau}\hat{z}(t-\tau)+Gy(t)+FBu(t-2\tau_u), 
\end{align}
where $\hat{z}(\theta)=\rho(\theta)$  for $\theta\in[-\tau,0]$, and $\tau=\tau_y-\tau_u$.
The matrices $N$, $N_{\tau}$, and $G$ are to be determined so that
$\hat{z}(t)\to z(t)$ asymptotically. In this setup, the observer (\ref{b9}) acts as a time-delay compensator, leveraging delayed measurements to estimate more current functional.

Defining the estimation error vector $e(t)=\hat z(t)-Fx(t-\tau_u)$, the error dynamics are given by
\begin{align}
	\label{b10}
	\dot{e}(t)&=\dot{\hat{z}}(t)-F\dot{x}(t-\tau_u)=Ne(t)+N_{\tau}e(t-\tau)\nonumber\\&+\mathcal{C}_{1}x(t-\tau_u)+\mathcal{\tilde{C}}_{1}x(t-\tau_y),
\end{align}
where $\tau=\tau_y-\tau_u$, $\mathcal{C}_1 =NF-FA$ and $\mathcal{\tilde{C}}_1 =N_{\tau}F+GC$.

If $\begin{pmatrix} \mathcal {C}_1 & \mathcal{\tilde{C}}_{1}\end{pmatrix}=\bf 0$, equation \eqref{b10} reduces to \eqref{b8}, decoupling the error dynamics from both $x(\cdot)$ and $u(\cdot)$. If \eqref{b8} is also asymptotically stable, then $e(t)\to \bf 0$ as $t\to\infty$ for all admissible initial conditions and inputs $u(\cdot)$.

According to \cite{trinhnam26}, the constraint $\mathcal{C}_{1}=\mathbf{0}$ holds if and only if$$\text{rank} \begin{pmatrix} FA \\ F \end{pmatrix} = \text{rank}(F),$$in which case the matrix $N$ is given by $N=FAF^-$. Since $C$ has full rank ($\text{rank}(C)=n$), the second constraint, $\mathcal{\tilde{C}}_{1}=\mathbf{0}$, is satisfied by setting
\begin{align}
\label{b11}
G&=-N_{\tau}FC^{-}.\end{align}

Thus, for given values of $\tau_u$ and $\tau_y$, we obtain $\tau=\tau_y-\tau_u$, then Lemma 11 in \cite{trinhnam26} can be employed to determine $N_{\tau}$ such that system (\ref{b8}) is asymptotically stable. A sufficient condition for this stability is the feasibility of the LMI presented in Lemma 11 of \cite{trinhnam26}.

\textit{Remark 3:} By explicitly accounting for the time shift $\tau$ between the output and input vectors, the observers in \eqref{b5}-\eqref{b6} and \eqref{b9} provide less conservative results than those in \cite{trinhnn26}. Conversely, Lemma 11 in \cite{trinhnam26} ensures the asymptotic stability of \eqref{b8} for a given $\tau$ up to its maximum upper bound $\tau_M$; beyond this bound ($\tau > \tau_M$), stability can no longer be guaranteed by Lemma 11 alone. Note, however, that when stabilizing \eqref{b1} via the delayed control law \eqref{b3}, the LMI condition of Lemma 11 remains feasible for input delays up to $\bar{\tau}_u$. Thus, if $\tau > \tau_M>\bar{\tau}_u$ while $\tau_u \ll \bar{\tau}_u$, one can deliberately select a larger input delay $\tilde{\tau}_u$ such that $\tau_u < \tilde{\tau}_u < \bar{\tau}_u$. By implementing this intentionally further-delayed control law$$u(t-\tilde{\tau}_u)=Fx(t-\tilde{\tau}_u),$$system \eqref{b1} remains stabilized since the LMI condition is still satisfied. The trade-off is a compromise in control robustness due to the increased input delay. Nevertheless, this strategy effectively reduces the net shift to $\tilde{\tau}=\tau_y-\tilde{\tau}_u$, thereby restoring the asymptotic stability of the error dynamics \eqref{b8}.

\textit{Example 1:} Let us revisit Example 2 in \cite{trinhnn26} where we now let $\tau_u = 0.25\text{s}$, $\tau_y = 1.7\text{s}$, $C = I_2$ and
\[A=\begin{pmatrix} 0  &   1\\-1  & 1 \end{pmatrix}, \quad B=\begin{pmatrix} 0\\1 \end{pmatrix},\]
to design an observer-based controller that stabilizes the system. Firstly, note that $\tau=\tau_y-\tau_u=1.45\text{s}$. We obtain the following control law
$$u(t - 0.25) = Fx(t - 0.25),$$where the gain matrix$$F = \begin{pmatrix} 0.5309 &  -1.6112 \end{pmatrix}$$is obtained by applying Lemma 11 \cite{trinhnam26} (with $\tau_u = 0.25\text{s}$ and $\lambda = 1$).

Next, we employ the second-order observer \eqref{b9} to estimate the augmented functional$$z_{a}(t)=\bar{F}x(t-0.25),$$where$$\bar{F} = \begin{pmatrix} 0.5309 &  -1.6112\\0 &1 \end{pmatrix}.$$Estimating a second-order functional ensures that the rank condition$$\text{rank} \begin{pmatrix} \bar{F}A \\ \bar{F} \end{pmatrix} = \text{rank}(\bar{F})$$is satisfied. A more detailed discussion on this requirement can be found in \cite{trinhnn26, trinhnam26}.

However, given $\bar{N}=\bar{F}A\bar{F}^-=\begin{pmatrix} 3.0349 &   3.8096\\ -1.8836 &  -2.0349 \end{pmatrix}$ and $\tau=1.45$s, the LMI condition of Lemma 11 in \cite{trinhnam26} becomes infeasible. Numerical analysis reveals that this specific LMI condition remains feasible only for delays up to the maximum upper bound of $\tau_M = 1.26$s. 

Therefore, as discussed in Remark 3,  we now add delay to the control input ($\tau_u \to \tilde{\tau}_u$) to fix a stability issue in the observer error dynamics. Since, for this example, the LMI condition presented in Lemma 11 \cite{trinhnam26} remains feasible for input delays up to $\bar{\tau}_u = 0.95\text{s}$. Let us now, pick $\tilde{\tau}_u$ as, say, $\tilde{\tau}_u=0.75\text{s}$ so that $\tilde{\tau}=\tau_y-\tilde{\tau}_u=0.95\text{s}$ which is under the safe threshold $\tau_M$.

Accordingly, we design a delayed control law of the form$$u(t - 0.75) = Fx(t - 0.75),$$where the gain matrix$$F = \begin{pmatrix} 0.8121 &  -0.8469 \end{pmatrix}$$is obtained by applying Lemma 11 in \cite{trinhnam26} with $\tilde{\tau}_u = 0.75\text{s}$ and $\lambda = 1$.

This approach represents a deliberate design trade-off. While artificially increasing the input delay slows the controller response and inherently degrades the robustness of the state stabilization, it effectively mitigates the divergence of the observer error dynamics. Consequently, this strategy forces the net delay into a stable regime, ensuring that the state estimates converge asymptotically to their true values.

Next, we employ a second-order observer (\ref{b9}) to estimate the following augmented functional
\[z_{a}(t)=\bar{F}x(t-0.75),\]
where
$$\bar{F} = \begin{pmatrix} 0.8121 &  -0.8469\\0 &1 \end{pmatrix}.$$
Now, given $\bar{N}=\bar{F}A\bar{F}^-=\begin{pmatrix} 1.0429  &  0.8484\\
	-1.2314  & -0.0429 \end{pmatrix}$ and $\tilde{\tau}=0.95$s, the LMI condition of Lemma 11 in \cite{trinhnam26}  is feasible and we obtain
\[\bar{N}_{\tau}=\begin{pmatrix}-0.7925 &  -0.4862\\0.7561 &  -0.2975\end{pmatrix}\] and that the below error time-delay system is asymptotically stable
\begin{align}
	\dot{e}(t) = \bar{N}e(t) + \bar{N}_{\tau}e(t - 0.95)\nonumber.
\end{align}
From (\ref{b11}), we obtain $G$,  and the functional $z(t)=u(t-0.75)=Fx(t-0.75)$ can now be estimated according to the following equations
\begin{align}
	\hat{z}(t)& = \begin{pmatrix} 1 & 0\end{pmatrix} \hat{z}_a(t)\nonumber\\
	\dot{\hat{z}}_a(t)&=\begin{pmatrix} 1.0429  &  0.8484\\
		-1.2314  & -0.0429 \end{pmatrix}\hat{z}_a(t)\nonumber\\&+\begin{pmatrix}-0.7925 &  -0.4862\\0.7561 &  -0.2975\end{pmatrix}\hat{z}_a(t-0.95) \nonumber\\&+\begin{pmatrix} 0.6436 &  -0.1851\\
		-0.6140 &   0.9378 \end{pmatrix}y(t)+\begin{pmatrix} -0.8469\\
		1\end{pmatrix}u(t-1.5).\nonumber
\end{align}
To illustrate the validity of the observer-based control scheme which combines the controller and observer derived above, Figure \ref{fig1paper4} shows the trajectories of $x_1(t)$ and  $x_2(t)$. It is clear
that asymptotic stability of the closed-loop system has been achieved.
\begin{figure}[!h]
	\centering
	\includegraphics[width=0.9\linewidth]{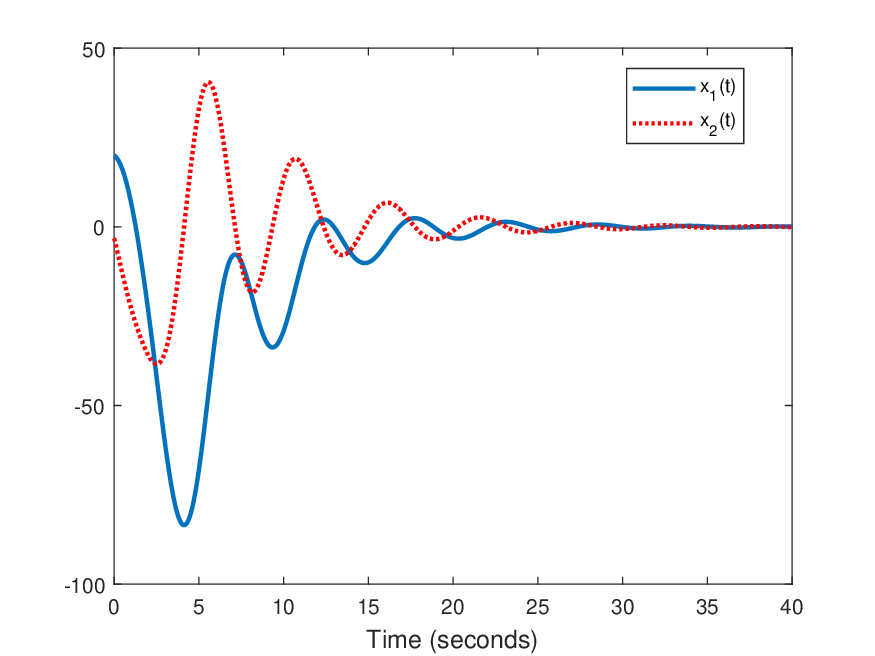}
	\caption{Trajectories of $x_1(t)$ and $x_2(t)$}
	\label{fig1paper4}
\end{figure}

Building upon the foundational work in \cite{trinhnam26, trinhnam1}, we introduce a functional observer incorporating multiple internal delays and delayed output measurements to expand $\tau$ and improve the robustness of the error-estimation time-delay system. We first consider the case where $C$ has full row rank; the scenario where $\text{rank}(C)=n$ can be subsequently deduced from this result (see Remark 5).

Consider the following observer
\begin{align}
	\label{b12}
	&\hat{z}(t)= w(t)+My(t)+M_{\alpha}y(t-\alpha),\\
	\label{b13}
	&\dot{w}(t)=Nw(t)+N_{\tau}w(t-\tau)+N_{\alpha \tau}w(t-\alpha -\tau)\nonumber\\&+Gy(t)+G_{\alpha}y(t-\alpha)+G_{\tau}y(t-\tau)+G_{\alpha \tau}y(t-\alpha-\tau) \nonumber\\
	&+G_{2\alpha \tau}y(t-2\alpha-\tau)+ Ju(t-2\tau_u)+J_1u(t-\tau_u-\tau_y)\nonumber\\
	&+J_2u(t-\alpha-\tau_u-\tau_y), 
\end{align}
where $w(\theta)=\rho(\theta)$  for $\theta\in[-(\tau+\alpha),0]$, $w(t)\in\mathbb{R}^r$, $\tau=\tau_y-\tau_u>0$, and $\alpha >0$ is a parameter
to be specified.
The matrices $M$, $M_{\alpha}$, $N$, $N_{\tau}$, $N_{\alpha \tau}$, $G$, $G_{\alpha}$, $G_{\tau}$, $G_{\alpha \tau}$, $G_{2\alpha \tau}$, $J$, $J_1$ and $J_2$ are design parameters to be determined such that $\hat{z}(t)\to z(t)$ asymptotically. 

Defining the estimation error vector $e(t)=\hat z(t)-z(t)$, the error dynamics are given by
\begin{align}
	\label{b14}
	\dot{e}(t)&=\dot{w}(t)+M\dot{y}(t)+M_{\alpha}\dot{y}(t-\alpha)-F\dot{x}(t-\tau_u)\nonumber\\ &=Ne(t)+N_{\tau}e(t-\tau)+N_{\alpha \tau}e(t-\alpha -\tau)\nonumber\\ &+\mathcal{C}_{1}x(t-\tau_u)+\mathcal{C}_{2}x(t-\tau_y)+\mathcal{\bar{C}}_{3}x(t-\alpha-\tau_y)\nonumber\\ &+\mathcal{\bar{C}}_{4}x(t-\tau -\tau_y)+\mathcal{\bar{C}}_{5}x(t-\alpha -\tau -\tau_y)\nonumber\\ &+\mathcal{\bar{C}}_{6}x(t-2\alpha -\tau -\tau_y)+\mathcal{\bar{C}}_{7}u(t-2\tau_u)\nonumber\\&+\mathcal{\bar{C}}_{8}u(t-\tau_u-\tau_y)+\mathcal{\bar{C}}_{9}u(t-\alpha -\tau_u-\tau_y),
\end{align}
where\\ 
$\mathcal{C}_1 =NF-FA$, \ $\mathcal{C}_2 =N_{\tau}F+\bar{G}C+MCA$, \ $\mathcal{\bar{C}}_{3}=N_{\alpha \tau}F+\bar{G}_{\alpha}C+M_{\alpha}CA$, \ $\mathcal{\bar{C}}_{4}=\bar{G}_{\tau}C$, \ $\mathcal{\bar{C}}_{5}=\bar{G}_{\alpha \tau}C$, \ $\mathcal{\bar{C}}_{6}=\bar{G}_{2\alpha \tau}C$,  \ $\mathcal{\bar{C}}_{7}=J-FB$, \ $\mathcal{\bar{C}}_{8} =J_1+MCB$, \ $\mathcal{\bar{C}}_{9} =J_2+M_{\alpha}CB$, \ $\bar{G}:=G-NM$, \ $\bar{G}_{\alpha}:=G_{\alpha}-NM_{\alpha}$, \ $\bar{G}_{\tau}:=G_{\tau}-N_{\tau}M$, \ $\bar{G}_{\alpha \tau}:=G_{\alpha \tau}-N_{\alpha \tau}M-N_{\tau}M_{\alpha}$, \ $\bar{G}_{2\alpha \tau}:=G_{2\alpha \tau}-N_{\alpha \tau}M_{\alpha}$ and $\mathcal{\bar{C}}=\begin{pmatrix}
	\mathcal {C}_1 &
	\mathcal {C}_2 & \mathcal{\bar{C}}_{3} &\mathcal{\bar{C}}_{4} &\mathcal{\bar{C}}_{5} &\mathcal{\bar{C}}_{6} &\mathcal{\bar{C}}_{7} &\mathcal{\bar{C}}_{8}&\mathcal{\bar{C}}_{9}
\end{pmatrix}$.

The following theorem characterizes sufficient conditions for the existence of observer (\ref{b12})-(\ref{b13}).

\begin{thm}\label{thm:1p3}
	For $\tau=\tau_y-\tau_u>0$, and $\alpha >0$, observer (\ref{b12})-(\ref{b13}) provides
	asymptotic estimation of the functional $z(t)=Fx(t-\tau_u)$ 
	if $\mathcal{\bar{C}}=\bf 0$ and the following delay-dependent error
	dynamics
	\begin{align}
		\label{b15}&\dot{e}(t)=Ne(t)+N_{\tau}e(t-\tau)+N_{\alpha \tau}e(t-\alpha -\tau)
	\end{align}
	is asymptotically stable. 
\end{thm}
\begin{proof} If $\mathcal{\bar{C}}=\bf 0$, then \eqref{b14} reduces to \eqref{b15}, so the error dynamics are
decoupled from $x(\cdot)$ and $u(\cdot)$. If, in addition, \eqref{b15} is asymptotically stable, then $e(t)\to \bf 0$ as $t\to\infty$ for all admissible initial conditions and inputs $u(\cdot)$. This completes the proof.
\end{proof}

\textit{Remark 4:} The primary purpose of utilizing the observer (\ref{b12})-(\ref{b13}) is to maximize the allowable time delay $\tau$ and enhance the robustness of the error estimation system (\ref{b15}). As reported in \cite{trinhnam26, trinhnam1}, the asymptotic stability conditions for (\ref{b15}) are less conservative than those for (\ref{b8}) due to the application of Lemma 13 in \cite{trinhnam26}. Consequently, system (\ref{b15}) maintains stability under a larger time delay $\tau$ than system (\ref{b8}). Specifically, by strategically leveraging the delayed error term $N_{\alpha \tau}e(t-\alpha -\tau)$, asymptotic stability for system (\ref{b15}) can be achieved even when the delay $\tau$ exceeds the upper bound $\tau_M$. Furthermore, the parameter $\alpha$ acts as a tuning variable. When integrated with the LMI condition in Lemma 13 in \cite{trinhnam26}, $\alpha$ can be optimized to yield a faster error convergence rate (see Example 2). Ultimately, the proposed observer (\ref{b12})-(\ref{b13}) represents a deliberate design trade-off, balancing improved estimation robustness against increased complexity from additional design parameters. Overall, this framework provides designers with a flexible range of options to address the complex issue of estimating delayed functionals.

In the following, we outline a procedure for finding the observer parameters that satisfy the existence conditions stated in Theorem \ref{thm:1p3}.

First, consider the constraints $\mathcal{\bar{C}}_7 = \mathcal{\bar{C}}_8 = \mathcal{\bar{C}}_9 = \mathbf{0}$, which are directly satisfied by choosing $J$, $J_1$, and $J_2$ as follows
\begin{align}
	\label{b16}
	&J = FB, \ J_{1}=-MCB, \ J_2=-M_{\alpha}CB.
\end{align}

Next, since $C$ has full row rank, the constraints $\mathcal{\bar{C}}_4 = \mathcal{\bar{C}}_5 = \mathcal{\bar{C}}_6 = \mathbf{0}$ are satisfied if and only if $\bar{G}_{\tau} = \bar{G}_{\alpha \tau} = \bar{G}_{2\alpha \tau} = \mathbf{0}$. Consequently, $G_{\tau}$, $G_{\alpha \tau}$, and $G_{2\alpha \tau}$ can be determined as follows
\begin{align}
	\label{b17}
	&G_{\tau}=N_{\tau}M, \ G_{\alpha \tau}=N_{\alpha \tau}M+N_{\tau}M_{\alpha}, \ G_{2\alpha \tau}=N_{\alpha \tau}M_{\alpha}.
\end{align}
The constraint $\mathcal{C}_{1}=\mathbf{0}$ holds if and only if$$\text{rank} \begin{pmatrix} FA \\ F \end{pmatrix} = \text{rank}(F),$$in which case the matrix $N$ is given by $N=FAF^-$.

As for
$\mathcal{C}_2 = \mathcal{\bar{C}}_3=\bf 0$, they can be expressed as follows
\begin{align}
	\label{b18}
	&\begin{pmatrix}N_{\tau} &\bar{G} &M\end{pmatrix}\begin{pmatrix}F\\C\\CA\end{pmatrix}=\bf 0,\\
	\label{b19}
	&\begin{pmatrix}N_{\alpha \tau} &\bar{G}_{\alpha} &M_{\alpha}\end{pmatrix}\begin{pmatrix}F\\C\\CA\end{pmatrix}=\bf 0.
\end{align}

By Lemma 2 in \cite{trinhnam26}, equations (\ref{b18}) and (\ref{b19}) admit a non-trivial solution if and only if an orthogonal basis exists for the left null-space of$$\begin{pmatrix}F\\C\\CA\end{pmatrix}.$$

The reader is referred to \cite{trinhnam26} for a comprehensive analysis of the solutions to (\ref{b18}) and (\ref{b19}), including the specific case where $\text{rank} \begin{pmatrix} C\\CA \end{pmatrix} = n$. Ultimately, these solutions ensure the asymptotic stability of system (\ref{b15}), subject to the feasibility of the LMI condition in Lemma 13 of \cite{trinhnam26}.

\textit{Remark 5:} When $\text{rank}(C) = n$, a functional observer with fewer parameters than (\ref{b12})-(\ref{b13}) can be constructed as follows
\begin{align}
	\label{b20}
	\dot{\hat{z}}(t)&=N\hat{z}(t)+N_{\tau}\hat{z}(t-\tau)+N_{\alpha \tau}\hat{z}(t-\alpha -\tau)\nonumber\\&+Gy(t)+G_{\alpha}y(t-\alpha)+FBu(t-2\tau_u), 
\end{align}
where $\hat{z}(\theta)=\rho(\theta)$  for $\theta\in[-(\tau+\alpha),0]$, $\tau=\tau_y-\tau_u$ and $\alpha >0$.
The matrices $N$, $N_{\tau}$, $N_{\alpha \tau}$, $G$ and $G_{\alpha}$ are to be determined so that
$\hat{z}(t)\to z(t)$ asymptotically.

The following corollary, derived from Theorem 1, establishes the existence conditions for the observer in (\ref{b20}).
\begin{corol}
	For $\tau=\tau_y-\tau_u>0$, and $\alpha >0$, observer (\ref{b20}) provides
	asymptotic estimation of the functional $z(t)=Fx(t-\tau_u)$ 
	if 
\begin{align}
	\label{b21}
	NF-FA&=\bf 0,\\
	\label{b22}
	N_{\tau}F+GC&=\bf 0,\\
	\label{b23}
	N_{\alpha \tau}F+G_{\alpha}C&=\bf0,
\end{align}
and the delay-dependent error
system (\ref{b15}) 
	is asymptotically stable.
\end{corol}

Note that (\ref{b21}) corresponds to the constraint $\mathcal{C}_{1}=\mathbf{0}$, which has already been discussed. With $\text{rank}(C) = n$, solutions to (\ref{b22}) and (\ref{b23}) are given, respectively, by (\ref{b11}) and
\begin{align}
	\label{b24}
	G_{\alpha}&=-N_{\alpha \tau}FC^{-}.\end{align}

Thus, given $\tau_u$ and $\tau_y$, we compute $\tau = \tau_y - \tau_u$. Then, for a specified $\alpha > 0$, the LMI condition in Lemma 13 of \cite{trinhnam26} is employed to determine $N_{\alpha \tau}$ and guarantee the asymptotic stability of the error system
\begin{align}
	\label{b25}\dot{e}(t)=FAF^-e(t)+N_{\tau}e(t-\tau)+N_{\alpha \tau}e(t-\alpha -\tau).\end{align}
As discussed in Remark 4, $\alpha$ serves as a tuning parameter that can be optimized alongside the LMI condition to maximize the error convergence rate.

\textit{Example 2:} We revisit Example 1, this time employing observer (\ref{b20}) in place of observer (\ref{b9}). Recall that by selecting $\tilde{\tau}=0.95$s, the LMI condition of Lemma 11 in \cite{trinhnam26}  is feasible. 

To illustrate the role of $\alpha$ as a tuning parameter for accelerating the error convergence rate, we solve the LMI in Lemma 13 of \cite{trinhnam26} with $\tilde{\tau}=0.95$ for $\alpha =0.55$ and $\alpha =0.95$. In both instances, the LMI is feasible for $\lambda=1$, yielding the corresponding matrices $N_{\tau}$ and $N_{\alpha \tau}$. To demonstrate that observer (\ref{b20}) achieves faster error convergence than observer (\ref{b9}), Figure \ref{fig2paper4} plots the trajectories of $e_1(t)$ under both values of $\alpha$ alongside the trajectory obtained using observer (\ref{b9}), assuming identical initial conditions.
\begin{figure}[!h]
	\centering
	\includegraphics[width=0.9\linewidth]{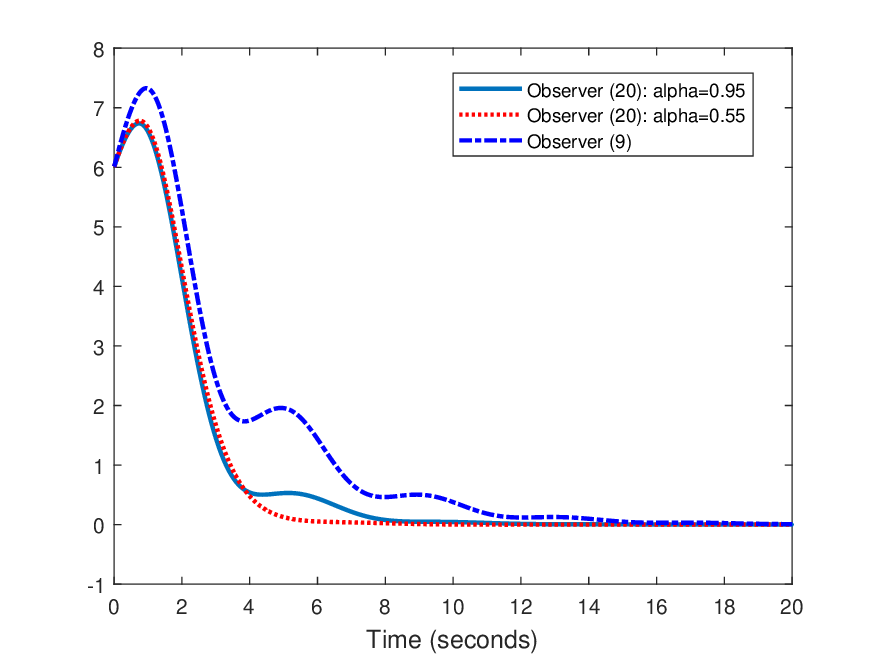}
	\caption{Trajectories of $e_1(t)$}
	\label{fig2paper4}
\end{figure}

Furthermore, as noted in Example 1, the LMI condition in Lemma 11 of \cite{trinhnam26} remains feasible for delays up to the maximum upper bound of $\tau_M = 1.26$s. In contrast, for $\tilde{\tau}=1.48$s and $\tilde{\tau}+\alpha=2$s, the LMI condition in Lemma 13 of \cite{trinhnam26} remains feasible, thereby ensuring the asymptotic stability of (\ref{b15}).

\textit{Example 3:}  This example is provided to further illustrate the practical implications of Remarks 3 and 4, specifically demonstrating how the observer (\ref{b20}) achieves a faster error convergence rate and an enhanced observer-based closed-loop response. For context, recall that in Example 1, selecting $\tilde{\tau}_u = 0.75\text{s}$ resulted in a relatively sluggish and highly oscillatory system response (see Figure \ref{fig1paper4}). Conversely, as established in Example 2, $\tilde{\tau}$ can be selected larger than $\tau_M$. To contrast these approaches, let us now choose a smaller input delay parameter, namely $\tilde{\tau}_u = 0.5\text{s}$. Note that with this choice, the shifting parameter becomes $\tilde{\tau} = \tau_y - \tilde{\tau}_u = 1.7 - 0.5 = 1.2\text{s}$. Because this value is now very close to the upper bound $\tau_M$, utilizing observer (\ref{b9}) would consequently result in a highly degraded and slow estimation error convergence rate. To overcome this limitation, we now employ the proposed observer (\ref{b20}) to estimate the delayed control law.

First, with $\tilde{\tau}=0.5\text{s}$, we design a delayed control law of the form$$u(t - 0.5) = Fx(t - 0.5),$$where the gain matrix$$F = \begin{pmatrix} 0.6979 & -1.2628 \end{pmatrix}$$is obtained by applying Lemma 11 in \cite{trinhnam26} with $\tilde{\tau}_u = 0.5\text{s}$ and $\lambda = 1$.

Next, we employ a second-order observer (\ref{b20}) to estimate the following augmented functional
\[z_{a}(t)=\bar{F}x(t-0.5),\]
where
$$\bar{F} = \begin{pmatrix} 0.6979 & -1.2628\\0 &1 \end{pmatrix}.$$

Given the matrix $\bar{N} = \bar{F}A\bar{F}^- = \begin{pmatrix} 1.8095 &   1.7202 \\ -1.4329 &  -0.8095 \end{pmatrix}$, alongside the delay parameter $\tilde{\tau} = 1.2\text{s}$ and $\alpha = 0.6$, the LMI condition formulated in Lemma 13 of \cite{trinhnam26} is found to be feasible with $\lambda = 1$. Consequently, we obtain the following observer gain matrices\[\bar{N}_{\tau}=\begin{pmatrix}-1.4078  & -1.2580\\1.3524  &  0.5868\end{pmatrix}, \bar{N}_{\alpha \tau}=\begin{pmatrix}0.1113 &   0.2171\\
	-0.4275 &  -0.2879\end{pmatrix}.\]
This feasibility guarantees that the estimation error system (\ref{b25}) is asymptotically stable, achieving an error convergence rate comparable to that of the closed-loop controller.

From (\ref{b11}) and (\ref{b24}), we obtain $G$ and $G_{\alpha}$, respectively. As a result,  the functional $z(t)=u(t-0.5)=Fx(t-0.5)$ can now be estimated according to the following equations
\begin{align}
	\hat{z}(t)& = \begin{pmatrix} 1 & 0\end{pmatrix} \hat{z}_a(t)\nonumber\\
	\dot{\hat{z}}_a(t)&= \begin{pmatrix} 1.8095 &   1.7202 \\ -1.4329 &  -0.8095 \end{pmatrix}\hat{z}_a(t)\nonumber\\&+\begin{pmatrix}-1.4078  & -1.2580\\1.3524  &  0.5868\end{pmatrix}\hat{z}_a(t-1.2) \nonumber\\&+\begin{pmatrix}0.1113 &   0.2171\\
		-0.4275 &  -0.2879\end{pmatrix}\hat{z}_a(t-1.8)\nonumber\\&+\begin{pmatrix}0.9825 &  -0.5197\\-0.9439 &  1.1211\end{pmatrix}y(t)\nonumber\\&+\begin{pmatrix}-0.0777 &  -0.0766\\
		0.2984 &  -0.2520\end{pmatrix}y(t-0.6)\nonumber\\&+\begin{pmatrix}-1.2628\\
		1\end{pmatrix} u(t-1).\nonumber
\end{align}
To demonstrate the effectiveness of the proposed observer-based control scheme, which integrates the controller and observer derived above, the state trajectories of $x_1(t)$ and $x_2(t)$ are illustrated in Figure \ref{fig3paper4}. The simulation results clearly indicate that both asymptotic stability of the closed-loop system and an enhanced convergence rate are successfully achieved (cf. Figure \ref{fig1paper4}).
\begin{figure}[!h]
	\centering
	\includegraphics[width=0.9\linewidth]{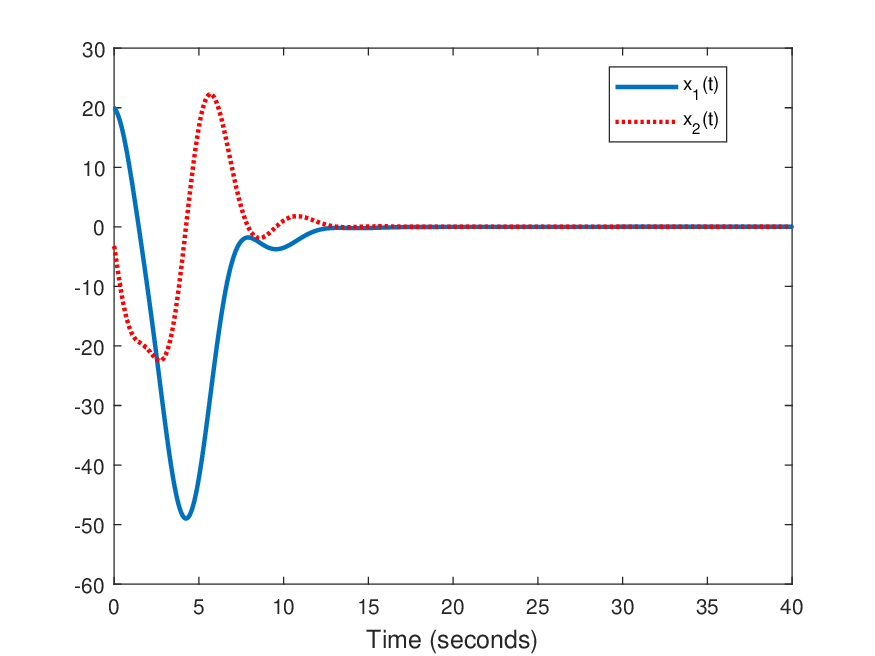}
	\caption{Trajectories of $x_1(t)$ and $x_2(t)$}
	\label{fig3paper4}
\end{figure}

\section{Conclusion}

In this paper, a novel class of delayed functional observers has been successfully developed to reconstruct delayed control laws subject to severe output measurement lags, thereby extending the foundational results reported in \cite{trinhnn26, trinhnam26}. The core strength of the proposed architecture lies in its ability to systematically mitigate simultaneous, unequal delays across both the actuation and sensing channels. By aligning the mismatched system timelines, the framework resolves dual-channel latency without the need for full-state estimation or computationally intensive, real-time distributed integration. Consequently, the collective advancements presented in this work and related literature \cite{trinhnn26, trinhnam26, trinhnam1} establish a highly efficient, low-order framework that successfully bridges the gap between idealized, delay-free control theory and the practical communication constraints of modern networked engineering systems.

\end{document}